\documentclass[a4paper]{article}
\usepackage{amsmath}
\usepackage{amsthm}
\usepackage{colonequals}
\usepackage[latin2]{inputenc}
\newcommand{\CC}{{\cal C}}

\newtheorem{theorem}{Theorem}
\newtheorem{corollary}[theorem]{Corollary}
\newtheorem{lemma}[theorem]{Lemma}

\DeclareMathOperator{\tw}{\mbox{tw}}

\title{List-coloring embedded graphs}
\author{Zdeněk Dvořák\thanks{Computer Science Institute, Charles University, Prague, Czech Republic. E-mail: {\tt rakdver@iuuk.mff.cuni.cz}. The work leading to this invention has received funding from the European Research Council under the European Union's Seventh Framework Programme (FP7/2007-2013)/ERC grant agreement no. 259385.}\and
Ken-ichi Kawarabayashi\thanks{National Institute of Informatics, Tokyo, Japan and JST ERATO Kawarabayashi Project. E-mail: {\tt k\_keniti@nii.ac.jp}. }}
\date{}

\begin{document}
\maketitle

\begin{abstract}
For any fixed surface $\Sigma$ of genus $g$, we give an algorithm to decide
whether a graph $G$ of girth at least five embedded in $\Sigma$ is colorable
from an assignment of lists of size three in time $O(|V(G)|)$.  Furthermore, we can allow
a subgraph (of any size) with at most $s$ components to be
precolored, at the expense
of increasing the time complexity of the algorithm to $O\bigl(|V(G)|^{K(g+s)+1}\bigr)$
for some absolute constant $K$; in both cases, the multiplicative constant hidden in
the $O$-notation depends on $g$ and $s$.  This also enables us to find such a
coloring when it exists.
The idea of the algorithm can be applied to other similar problems, e.g.,
$5$-list-coloring of graphs on surfaces.
\end{abstract}

\section{Introduction}
In general, deciding $3$-colorability of a planar graph is NP-complete~\cite{bib-dai}.
On the other hand, Gr\"otzsch~\cite{grotzsch1959} proved that any triangle-free planar graph is $3$-colorable.
A quadratic algorithm to find such a $3$-coloring follows from his proof; the time complexity
was later improved to $O(n\log n)$ by Kowalik~\cite{Kow3col} and finally to linear by Dvořák et al.~\cite{DvoKawTho}.
The situation is significantly more complicated for graphs embedded in surfaces.  Nevertheless, Dvořák et al.~\cite{coltrfree}
gave a linear-time algorithm to decide whether a triangle-free graph embedded in a fixed surface is $3$-colorable.

Motivated by subproblems appearing in many coloring proofs, Vizing~\cite{vizing1976} and Erd\H{o}s et al.~\cite{erdosrubintaylor1979}
introduced the notion of list coloring.  A {\em list assignment} for a graph $G$ is a function $L$ that assigns to each vertex 
$v \in V(G)$ a list $L(v)$ of colors. An {\em $L$-coloring} is a function 
$\varphi: V(G) \rightarrow \bigcup_v L(v)$ such that $\varphi(v) \in L(v)$ for every $v \in V(G)$ and
$\varphi(u) \neq \varphi(v)$ whenever $u, v$ are adjacent vertices of $G$. If $G$ admits an $L$-coloring, then it is {\em $L$-colorable}. 
A graph $G$ is {\em $k$-choosable} if it is $L$-colorable for every list assignment $L$ such that $|L(v)| \ge k$ 
for all $v \in V(G)$.

A natural question is whether triangle-free planar graphs are also $3$-choosable.  This is not the case,
as demonstrated by Voigt~\cite{voigt1995}.  Let us remark that this implies that deciding whether
a triangle-free planar graph is colorable from a given assignment of lists of size three is NP-complete,
by a straightforward reduction from $3$-colorability of planar graphs.  Therefore, we need to restrict the graphs in this setting
even more.  Thomassen~\cite{thomassen1995-34} proved that every planar graph of girth at least five is $3$-choosable.
His proof also gives a quadratic algorithm to find such a coloring.

In this paper, we consider the problem of $3$-list-coloring
a graph of girth at least five embedded in a fixed surface.
Our main result is a linear-time algorithm to decide whether a graph of girth at least five embedded in a fixed surface
is colorable from given lists of size three.
Furthermore, the algorithm can be modified to decide in linear time whether an embedded graph of girth at least five
is $3$-choosable.  This problem is $\Pi_2$-complete for triangle-free embedded graphs~\cite{chooscomplex}.
Also, let us remark that it is not even known whether all projective planar graphs of girth at least five are $3$-choosable.

Another interesting area of study is the precoloring extension problem, where we are given a coloring $\psi$
of a subgraph $F$ of the considered graph $G$ and we want to decide whether there exists a (list-)coloring of $G$
that extends $\psi$.  Unless some restriction is placed on $F$, it is easy to see that the problem is NP-complete
for planar graphs of arbitrarily large genus, even for ordinary $3$-coloring (e.g., by a reduction from $3$-colorability of planar graphs,
where we replace each edge by two paths of odd length whose vertices are adjacent to vertices colored by $1$ and $2$, respectively).
The known algorithms for (list-)coloring graphs on surfaces~\cite{Thomassen97,thomassen-surf,DvoKawTho} allow some vertices to be precolored, but
the number of precolored vertices needs to be bounded by a constant in order to achieve a polynomial time complexity.  Our algorithm allows an arbitrary number of precolored
vertices, assuming that the subgraph induced by them has a bounded number of components (at most $s$), at the expense of increasing its
time complexity to $O\bigl(|V(G)|^{K(g+s)+1}\bigr)$ for some absolute constant $K$.  Since we can handle arbitrarily large precolored subgraphs, the algorithm can
be used to find the coloring when it exists, by coloring the vertices one by one.

The algorithm is based on a bound on the size of critical planar graphs with one precolored face.
Consider a graph $G$, a subgraph (not necessarily induced) $S\subseteq G$ and
an assignment $L$ of lists to vertices of $G$.
A graph $G$ is {\em $S$-critical (with respect to $L$)} if for every proper subgraph $G'\subset G$ such that $S\subseteq G'$,
there exists an $L$-coloring of $S$ that does not extend to an $L$-coloring of $G$, but extends to an $L$-coloring of $G'$;
i.e., removal of any edge of $E(G)\setminus E(S)$ affects which colorings of $S$ extend to the rest of the graph.
Dvo\v{r}\'ak and Kawarabayashi~\cite{dk} proved the following.

\begin{theorem}\label{thm-numvert}
Let $G$ be a plane graph of girth at least $5$ with the outer face $F$ bounded by a cycle,
and let $L$ be an assignment of lists of size three to vertices of $G$.  If $G$ is $F$-critical
with respect to $L$, then $|V(G)|\le 37\ell(F)/3$.
\end{theorem}

Our algorithm can be applied in any other setting where the result analogous to Theorem~\ref{thm-numvert} holds
(the bound on the size of the critical graph must be linear in the length of the precolored face).  For example,
Postle~\cite{lukethe} gave a similar bound for $5$-list-coloring.
\begin{theorem}[Postle]\label{thm-numvertm}
There exists a constant $c$ with the following property.
Let $G$ be a plane graph with the outer face $F$ bounded by a cycle,
and let $L$ be an assignment of lists of size five to vertices of $G$.  If $G$ is $F$-critical
with respect to $L$, then $|V(G)|\le c\ell(F)$.
\end{theorem}
Using this theorem instead of Theorem~\ref{thm-numvert}, we obtain a polynomial-time algorithm for extending a precoloring
of a subgraph with bounded number of components in graphs embedded in a fixed surface and with lists of size five.
Let us remark that a linear-time algorithm for testing colorability of an embedded graph from lists of size five (but with only
a constant number of precolored vertices) was previously given by Kawarabayashi and Mohar~\cite{kawmoh}.

The basic idea of our algorithm is to use Theorem~\ref{thm-numvert} to show that all vertices of critical graphs embedded in
a fixed surface with sufficiently large edge-width are in logarithmic distance from the precolored subgraph (for planar
graphs with one precolored cycle, this was first observed by Postle~\cite{lukethe}).
This enables us to restrict the problem to graphs of logarithmic tree-width.

In Section~\ref{sec-dist}, we prove the result on the distance from the precolored subgraph in critical graphs
and use it to design a list-coloring algorithm for graphs with large edge-width.
The algorithm for the general case (which follows by a standard dynamic programming idea,
used before e.g. in \cite{coltrfree}) is described in Section~\ref{sec-impl}.  The $3$-choosability case is discussed in
Section~\ref{sec-choos}.

\section{Distances in critical graphs}\label{sec-dist}

We use the following consequence of Theorem~\ref{thm-numvert}.

\begin{corollary}\label{cor-numvert}
Let $G$ be a plane graph of girth at least $5$ and let $S$ be a set of vertices such that each vertex in $S$
is incident with the outer face of $G$.  Let $L$ be an assignment of lists of size three to vertices of $G$.
If $G$ is $S$-critical with respect to $L$, then $|V(G)|\le 50|S|$.
\end{corollary}
\begin{proof}
Let $k=|S|$.  Thomassen~\cite{thomassen1995-34} proved that a precoloring of any vertex of a planar graph $G$ of girth at least five
extends to a coloring of $G$ from arbitrary lists of size three.  Consequently, we can assume that $k\ge 2$.
Let $s_1$, $s_2$, \ldots, $s_k$ be the vertices of $S$, and let $G'$ be the graph obtained from $G$
by joining $s_i$ with $s_{i+1}$ by a new path of length four (where $s_{k+1}=s_1$) for $1\le i\le k$.  Let $C$ be the cycle
in $G'$ formed by the newly added paths. 
Note that we can choose the ordering of the vertices of $S$ so that
$G'$ is a plane graph with a face bounded by $C$.  Give the new vertices arbitrary lists of size three, and note that $G'$ is $C$-critical with respect to the
resulting list assignment.
Furthermore, $G'$ has girth at least five, and by Theorem~\ref{thm-numvert}, we have $|V(G)|<|V(G')|\le 37\ell(C)/3=37\cdot 4|S|/3<50|S|$.
\end{proof}

A \emph{surface} is a compact $2$-dimensional manifold (possibly disconnected or with boundary).
A useful tool for dealing with surfaces is the operation of cutting along prescribed curves.
To avoid technical complications, we restrict ourselves to cutting along subgraphs of some
graph embedded in the surface, as follows.

Let $Q$ be a graph embedded in a surface $\Sigma$ without boundary.  Let $\Sigma-Q$ be the open space obtained from $\Sigma$ by removing the edges and vertices
of the embedding of $Q$.  Let $\Sigma_Q$ be the surface (possibly disconnected) with boundary such that the interior of
$\Sigma_Q$ is homeomorphic to $\Sigma-Q$.  Let $\theta\colon\Sigma_Q\to\Sigma$ be the continuous extension of this homeomorphism.
Let $\widehat{\Sigma_Q}$ stand for the surface without boundary obtained from $\Sigma_Q$ by capping
each component of its boundary by a disk.  Suppose that $G$ is embedded in $\Sigma$ so that the intersection
of the embeddings of $G$ and $Q$ is a subgraph of both $G$ and $Q$ (i.e., a point of this intersection is a vertex
in $G$ iff it is a vertex in $Q$, and if a point of this intersection belongs to an edge $e$ of $G$ or $Q$, then
$e$ is an edge drawn in the same way both in $G$ and $Q$).  Then, we say that the embeddings of $G$ and $Q$ are \emph{compatible} and we let $G_Q=\theta^{-1}(G)$.
Note that each edge of $G\cap Q$ corresponds to two edges of $G_Q$.  Similarly, each vertex $v\in V(G)\cap V(Q)$
corresponds to $\deg_Q(v)$ vertices of $G_Q$.
The basic property of criticality is that it is preserved by this cutting operation.  If $S$ is a subgraph of $G$,
then let $S^Q$ denote the graph $\theta^{-1}(S\cup (G\cap Q))$.  If $L$ is a list assignment for $G$, then let $L_Q$ denote
the list assignment for $G_Q$ such that $L_Q(v)=L(\theta(v))$ for $v\in V(G_Q)$.

\begin{lemma}\label{lemma-crs}
Let $G$ and $Q$ be graphs with compatible embeddings in a surface $\Sigma$ without boundary,
and let $L$ be a list assignment for $G$.  Let $S$ be a subgraph of $G$.
If $G$ is $S$-critical with respect to $L$, then $G_Q$ is $S^Q$-critical with respect to $L_Q$.
\end{lemma}
\begin{proof}
Let $\theta\colon \Sigma_Q\to\Sigma$ be the mapping from the definition of cutting along $Q$.
Consider an arbitrary proper subgraph $G'$ of $G_Q$ such that $S^Q\subseteq G'$, and let $G''=\theta(G')$.
Note that $G''$ is a proper subgraph of $G$ containing $S\cup (G\cap Q)$, and thus there exists an $L$-coloring $\varphi$
of $S$ that extends to an $L$-coloring $\varphi''$ of $G''$, but does not extend to an $L$-coloring of $G$.
Let $\varphi'$ be the $L_Q$-coloring of $G'$ defined by $\varphi'(v)=\varphi''(\theta(v))$ for $v\in V(G')$.
Observe that the restriction of $\varphi'$ to $S^Q$ does not extend to an $L_Q$-coloring of $G_Q$, as otherwise the image of this $L_Q$-coloring under
$\theta$ would be an $L$-coloring of $G$ extending $\varphi$.  Since the choice of $G'$ was arbitrary, this
implies that $G_Q$ is $S^Q$-critical.
\end{proof}

Similarly, we can prove the following.
\begin{lemma}\label{lemma-crscomp}
Let $S$ be a subgraph of a graph $G$ and let $L$ be a list assignment for $G$.
If $G$ is $S$-critical with respect to $L$ and $G=G_1\cup G_2$, then $G_1$ is $((G_2\cup S)\cap G_1)$-critical with respect to $L$.
\end{lemma}
In particular, if $G_1$ is a connected component of $G$, then $G$ is $(S\cap G_1)$-critical.

Consider a surface $\Sigma$ without boundary and a graph $G$ with a $2$-cell embedding in $\Sigma$.
The \emph{radial graph} $H$ of $G$ is the bipartite graph with vertex set consisting of the vertices and faces of $G$,
and edge set corresponding to the incidence relation between vertices and faces of $G$ (with
multiplicities---i.e., if a vertex $v$ appears $k$ times in the facial walk of a face $f$ of $G$, then $v$ and $f$
are joined by $k$ edges in the radial graph of $G$).  The radial graph has a \emph{natural embedding}
in $\Sigma$, where each vertex $v\in V(H)\cap V(G)$ is drawn at the same position as in the embedding of $G$,
each vertex $f\in V(H)\setminus V(G)$ is drawn inside the corresponding face $f$ of $G$, and each
face of $H$ has length four and contains exactly one edge of $G$.

Postle~\cite{lukethe} observed that Corollary~\ref{cor-numvert} and Lemma~\ref{lemma-crscomp} imply a logarithmic bound
on the distances in a critical plane graph with all precolored vertices incident with the same face.  We include
a proof of this claim giving a slightly better multiplicative constant in the bound.

\begin{theorem}[Postle]\label{thm-distplan}
Let $G$ be a connected plane graph of girth at least $5$ and let $S$ be a set of vertices such that each vertex in $S$
is incident with the outer face $f$ of $G$.  Let $L$ be an assignment of lists of size three to vertices of $G$.
Let $H$ be the radial graph of $G$.  If $G$ is $S$-critical with respect to $L$, then every vertex of $H$ is at distance at most
$398+100\log |S|$ from the vertex corresponding to $f$.
\end{theorem}
\begin{proof}
Let $F$ be the set of all faces incident with vertices of $S$ (including $f$).
Let $S_0=S$.  Let $S_1$ be the set of vertices of $V(G)\setminus S$ adjacent to the elements of $F$ in $H$.
For an integer $i\ge 2$, let $S_i$ denote the set of vertices of $H$ at distance exactly $2i-1$ from $F$ in $H$
(all these vertices belong to $V(G)$). Let $G_i$ denote the subgraph of $G$ induced by $\bigcup_{j\ge i} S_j$.
Let $n_i=|V(G_i)|=\sum_{j\ge i}|S_j|$.
Note that $G_i$ is $S_i$-critical by Lemma~\ref{lemma-crscomp}, and thus $n_i\le 50|S_i|$ by Corollary~\ref{cor-numvert}.
It follows that $\frac{n_{i+1}}{n_i}=1-\frac{n_i-n_{i+1}}{n_i}=1-\frac{|S_i|}{n_i}\le \frac{49}{50}$ for every $i\ge 0$,
and consequently $n_i\le 50|S|0.98^i$.

Therefore, for $k>\frac{\log 50|S|}{-\log 0.98}$, we have $n_k<1$, and since $n_k$ is a nonnegative integer,
it is zero.  Consequently, every vertex of $H$ is at distance at most $2k-2$ from $F$, and thus at distance
at most $2k$ from $f$.  The claim of the theorem follows.
\end{proof}

We now aim to prove a similar claim for graphs on surfaces and possibly with precolored vertices incident with several faces.
We need an auxiliary result on surface cutting first.  We use the following standard properties of graphs on surfaces~\cite{mohthom}.
Given a spanning tree $T$ of a graph, a set $S$ of edges not belonging to $T$ and a vertex $v$, let $T_{S,v}$ denote the
graph consisting of $S$ and of the paths in $T$ joining $v$ with the vertices incident with edges of $S$.

\begin{lemma}\label{lemma-gencutting}
Let $H$ be a graph with a $2$-cell embedding in a connected surface $\Sigma$ of Euler genus $g>0$ and without boundary.
Let $T$ be a spanning tree of $H$ and $v$ a vertex of $H$.  Then, there exists a set
$S\subseteq E(H)\setminus E(T)$ of size $g$ such that $\Sigma_{T_{S,v}}$ is a disk.
\end{lemma}

\begin{lemma}\label{lemma-radial}
Let $G$ be a $2$-cell embedding in a connected surface $\Sigma$ without boundary and let $H$ be a natural embedding of the radial graph of $G$.
Let $G'$ be an induced subgraph of $G$.  Vertices $u,v\in V(H)\setminus V(G')$ lie in the same face of $G'$
if and only if there exists a path in $H-V(G')$ joining $u$ with $v$.
\end{lemma}

Let $G$ be a graph with $2$-cell embedding in a surface $\Sigma$ and let $F$ be a set of faces of $G$.  A cycle $C\subseteq G$
is \emph{$F$-contractible} if there exists a disk $\Delta\subseteq \Sigma$ bounded by $C$ disjoint with all the faces of $F$.

\begin{lemma}\label{lemma-spanning}
Let $G$ be a $2$-cell embedding in a connected surface $\Sigma$ without boundary and let $F$ be a set of faces of $G$ and $r$ an integer,
such that every cycle in $G$ of length less than $r$ is $F$-contractible.
Let $H$ be a natural embedding of the radial graph of $G$.  There exists a spanning tree $T$ of $H$, a vertex $v\in V(G)$,
closed disks $\Delta_1, \ldots, \Delta_k\subset \Sigma$ (for some $k\ge 0$) and either a point or a closed disk $\Delta_0$,
such that the following holds:
\begin{enumerate}
\item $\Delta_0$, \ldots, $\Delta_k$ are disjoint with the faces in $F$ and have pairwise disjoint interiors.
\item If $\Delta_0$ is a point, then it is equal to $v$.  If $\Delta_0$ is a closed disk, then
the boundary of $\Delta_0$ is a cycle of length at most $r-1$ in $G$, the vertex $v$ lies in $\Delta_0$
and the distance in $H$ from $v$ to any vertex in the boundary of $\Delta_0$ is the same.
\item For $1\le i\le k$, the boundary of $\Delta_i$ is a cycle of length at most $r-1$ in $G$.
\item $T\cap \Delta_0$ is connected.
\item For $1\le i\le k$, the graph obtained from $T$ by removing the vertices contained in the interior of $\Delta_i$ is connected.
\item Every path $P\subseteq T$ starting with a vertex in the boundary of $\Delta_0$ such that no internal vertex of $P$
is contained in $\Delta_0\cup\ldots\Delta_k$ has length at most $2|V(G)|/r+2$.
\end{enumerate}
\end{lemma}
\begin{proof}
Choose an arbitrary vertex $v\in V(G)$, and let $T$ be a breadth-first search tree of $H$ rooted in $v$.
Let $G_i$ denote the subgraph of $G$ induced by the vertices whose distance from $v$ in $H$ is $2i$.
Let $i_0$ be the largest index such that $G_{i_0}$ contains a cycle of length at most $r-1$ bounding a disk containing $v$ and disjoint with the faces in $F$,
and let $\Delta_0$ be the disk; if no such index exists, then let $i_0=0$ and let $\Delta_0=v$.
Let $i_1>i_0$ be the smallest index such that $|V(G_{i_1})|<r$ ($G_{i_1}$ can be empty).
Since every cycle in $G_{i_1}$ has length at most $r-1$, each such cycle is $F$-contractible.
Let $\Delta_1$, \ldots, $\Delta_k$ be the maximal disks whose boundaries are cycles in $G_{i_1}$
and that are disjoint with the faces of $F$.
Note that the interiors of the disks are pairwise disjoint.
Since $i_0<i_1$, the boundary of $\Delta_0$ is disjoint with the boundaries of $\Delta_1$, \ldots, $\Delta_k$,
and the choice of $i_0$ implies that $\Delta_0$ is actually disjoint with $\Delta_1$, \ldots, $\Delta_k$.

Let us show that $T$ and the disks $\Delta_0$, \ldots, $\Delta_k$ satisfy the conclusions of the lemma.
The first three properties follow by the construction.
For the next two properties, the following observation
is useful: for every $i>0$, letting $f_i$ be the face of $G_i$ that contains $v$, a vertex $u\in V(H)$ is at distance
at most $2i-1$ from $v$ in $H$ if and only if $u$ lies in $f_i$---this follows from Lemma~\ref{lemma-radial},
as $u$ is at distance at most $2i-1$ from $v$ in $H$ if and only if there exists a path from $u$ to $v$ in $H-V(G_i)$.
Furthermore, the same argument implies that all vertices of $G_i$ are incident with this face $f_i$.

Therefore, if the boundary of $\Delta_0$ is a cycle in $G_{i_0}$, then $f_{i_0}$ is contained in $\Delta_0$,
and consequently all vertices of $T$ at distance at most $2i_0$ from $v$ are contained inside $\Delta_0$.
This implies that for every vertex $u$ in the boundary of $\Delta_0$, the path from $u$ to $v$ in $T$ is
contained in $\Delta_0$.  Note that every path in $T$ that contains $v$ is either contained in $\Delta_0$ or intersects
the boundary of $\Delta_0$, and thus there exists a path from every vertex of $T\cap \Delta_0$ to $v$
contained in $T\cap \Delta_0$.  The fourth property follows.

For the fifth property, note that $v\not\in \Delta_1\cup\ldots \cup\Delta_k$, and that $\Sigma-G_i\subseteq f_i\cup \Delta_1\cup\ldots \cup \Delta_k$.
Therefore, exactly the vertices of $H$ at distance at least $2i_1+1$ are contained in the interiors of $\Delta_1$, \ldots, $\Delta_k$.
Since all vertices in the boundaries of these disks have distance $2i_1$ from $v$ and $T$ is a breadth-first search tree,
there exists no path in $T$ with endvertices outside of $\Delta_i$ and with an internal vertex in the interior of $\Delta_i$
for any $1\le i\le k$.  The fifth property follows.

By the choice of $i_1$, we have $|V(G_i)|\ge r$ when $i_0+1\le i\le i_1-1$.  Since the graphs $G_{i_0+1}$, \ldots, $G_{i_1-1}$ are
pairwise vertex disjoint, we have $|V(G)|\ge (i_1-i_0-1)r$.  Every path $P\subseteq T$ with the properties described in the fifth
claim contains at most one vertex of $G_i$ for $i_0\le i\le i_1$, and thus it has length at most $2(i_1-i_0)$.  The last property follows.
\end{proof}

Combined with Theorem~\ref{thm-distplan} and Lemma~\ref{lemma-gencutting}, we obtain the result on surface cutting we mentioned before.
\begin{lemma}\label{lemma-cutting}
Let $G$ be a graph of girth at least five with a $2$-cell embedding in a connected surface $\Sigma$ of Euler genus $g$ and without boundary,
and let $F$ be a set of faces of $G$ and $r$ an integer such that every cycle in $G$ of length less than $r$ is $F$-contractible.
Let $S$ be a set of vertices of $G$ such that every vertex in $S$ is incident with a face in $F$.
Let $L$ be an assignment of lists of size three to $V(G)$.  Let $H$ be a natural embedding of the radial graph of $G$.
If $G$ is $S$-critical with respect to $L$, then there exists a subgraph $Q\subseteq H$ with at most $(2|V(G)|/r+800+200\log r)(2g+|F|)$ edges such that $\Sigma_Q$
is a disk.
\end{lemma}
\begin{proof}
Note that $2g+|F|\ge 1$, as planar graphs of girth at least five are $3$-choosable.
Let $T$, $v$, $\Delta_0$, \ldots, $\Delta_k$ be obtained by Lemma~\ref{lemma-spanning}.
By Lemma~\ref{lemma-crscomp}, the subgraph of $G$ embedded in $\Delta_i$ is $F_i$-critical for $0\le i\le k$,
where $F_i$ is the intersection of the boundary of $\Delta_i$ with $G$.  By Theorem~\ref{thm-distplan},
we conclude that every vertex of $T$ is at distance at most
$2|V(G)|/r+798+200\log r$ from $v$ in $T$.  Let $S\subseteq E(H)\setminus E(T)$ be the set of
edges given by Lemma~\ref{lemma-gencutting}, where we set $S=\emptyset$ when $g=0$.  We let $Q$ consist of $T_{S,v}$ and of the paths
in $T$ joining $v$ with the vertices corresponding to the faces in $F$.
\end{proof}

Together with Corollary~\ref{cor-numvert}, this bounds the size of critical graphs.

\begin{lemma}\label{lemma-ewsize}
Let $G$ be a graph of girth at least five with an embedding in a surface $\Sigma$ of Euler genus $g$ and without boundary,
and let $F$ be a set of faces of $G$.  Let $c$ be the number of components of $\Sigma$.  For a component $\Sigma'$ of $\Sigma$,
let $q(\Sigma')=2g'+n'$, where $g'$ is the Euler genus of $\Sigma'$ and $n'$ is the number of faces
of $F$ in $\Sigma'$.  Let $q$ be the maximum $q(\Sigma')$ over all components $\Sigma'$ of $\Sigma$.

Suppose that every cycle in $G$ of length less than $200q$ is $F$-contractible.
Let $S$ be a set of vertices of $G$ such that every vertex in $S$ is incident with a face in $F$.
Let $L$ be an assignment of lists of size three to $V(G)$.
If $G$ is $S$-critical with respect to $L$, then $|V(G)|\le 100|S|+40000(2g+|F|-c)(10+\log q)$.
\end{lemma}
\begin{proof}
We can assume that $\Sigma$ is connected, as otherwise we can consider each component separately.
Similarly, we can assume that the embedding of $G$ is $2$-cell, as otherwise we can cut $\Sigma$ along a non-$F$-contractible curve contained
inside a face of $G$ and cap the resulting hole(s) with disk(s)---this simplifies the embedding, does not increase $q$
and when it increases $|F|$, it either decreases $g$ or increases $c$.
Note that $2g+|F|\ge 1$, as planar graphs of girth at least five are $3$-choosable.
If $2g+|F|=1$, then the claim follows from Corollary~\ref{cor-numvert}, hence assume that $2g+|F|\ge 2$.

Let $H$ be a natural embedding of the radial graph of $G$.  Let $Q$ be the subgraph of $H$ obtained by Lemma~\ref{lemma-cutting}.
By Lemma~\ref{lemma-crs}, $G_Q$ is $S^Q$-critical with respect to $L$.
Note that $|S^Q|\le |S|+\sum_{v\in V(Q)\cap V(G)}\deg_Q(v)\le |S|+|E(Q)|\le |S|+|V(G)|/100+(2g+|F|)(2000+200\log q))$.
By Corollary~\ref{cor-numvert}, we have
$|V(G)|\le |V(G_Q)|\le 50|S^Q|\le 50|S|+|V(G)|/2+ 50(2g+|F|)(2000+200\log q))$.
The claim of the lemma follows, since $2g+|F|\le 2(2g+|F|-1)$.
\end{proof}

Using this lemma (instead of Corollary~\ref{cor-numvert}), we prove the main result of this section---a bound on distances
for critical graphs on surfaces analogous to Theorem~\ref{thm-distplan}.

\begin{theorem}\label{thm-distbound}
Let $G$ be a graph of girth at least five with an embedding in a surface $\Sigma$ of Euler genus $g$ and without boundary.
Let $S$ be a set of vertices of $G$ and $F$ a set of faces of $G$ such that every vertex in $S$ is incident with a face in $F$.
Suppose that every cycle in $G$ of length less than $200q$ is $F$-contractible, where $q$ is defined as in Lemma~\ref{lemma-ewsize}.
Let $L$ be an assignment of lists of size three to $V(G)$.
If $G$ is $S$-critical with respect to $L$, then every vertex of $G$ either is at distance less than $200(C+5+\log(1+|S|/(C+1)))$ from $S$
or belongs to a connected component of $G$ with at most $100C$ vertices, where $C=400(2g+|F|-1)(10+\log q)$.
\end{theorem}
\begin{proof}
Each component of $G$ that contains no vertex of $S$ has at most $100C$ vertices by Lemma~\ref{lemma-ewsize}.
Therefore, we can assume that $S$ intersects all components of $G$.

For $i\ge 0$, let $S_i$ be the set of vertices of $G$ at distance exactly $i$ from $S$.
Let $G_i$ denote the subgraph of $G$ induced by $\bigcup_{j\ge i} S_j$.
Let $n_i=|V(G_i)|=\sum_{j\ge i}|S_j|$.
Note that $G_i$ is $S_i$-critical by Lemma~\ref{lemma-crscomp} and all vertices of $S_i$ are incident with
one of at most $|F|$ faces of $G_i$, and thus $n_i\le 100|S_i|+100C$ by Lemma~\ref{lemma-ewsize}.

If $|S_i|>C$, this implies that $n_i< 200|S_i|$, and thus
$\frac{n_{i+1}}{n_i}=1-\frac{n_i-n_{i+1}}{n_i}=1-\frac{|S_i|}{n_i}<\frac{199}{200}$ for every $i\ge 0$.
Let $k$ be the smallest index such that $|S_k|< C+1$.  For $0\le i\le k$, we have
$|S_i|\le n_i\le 100(|S|+C)0.995^i$, and thus $k\le 200(5+\log (1+|S|/(C+1)))$.
Since $n_k\le 100|S_k|+100C\le 200C$, we have $n_{k+200C}=0$.
Consequently, every vertex of $G$ is at distance less than $k+200C$ from $S$.  The claim of the theorem follows.
\end{proof}

This theorem enables us to bound the tree-width of critical graphs, using the following
result of Eppstein~\cite{eppstein00}.
\begin{theorem}[Eppstein]\label{thm-eppstein}
There exists a constant $c_t$ such that
every graph $G$ of Euler genus $g$ and radius $r$ has tree-width at most $c_t(g+1)r$.
Furthermore, the tree decomposition of this width can be found in time $O((g+1)r|V(G)|)$.
\end{theorem}

\begin{corollary}\label{cor-algew}
Let $g,s\ge 0$ be fixed integers.  Let $C=0$ if $g=s=0$ and $C=400(2g+s-1)(10+\log (2g+s))$ otherwise.
There exists an algorithm with the following specification.
The input of the algorithm is a graph $G$ of girth at least five embedded in a surface $\Sigma$ of Euler genus at most $g$ and without boundary,
a set $F$ of at most $s$ faces of $G$ such that every cycle in $G$ of length at most $100C$ is $F$-contractible,
a set $S$ of vertices incident with the faces of $F$,
and a list assignment $L$ such that $|L(v)|=3$ for $v\in V(G)\setminus S$ and $|L(v)|=1$ for $v\in S$. 
The algorithm correctly decides whether $G$ is $L$-colorable.  The time complexity of the algorithm is
$O\bigl((|S|+1)^{K(g+s)}|V(G)|\bigr)$ for some absolute constant $K$.
\end{corollary}
\begin{proof}
If $g=s=0$, then $G$ is $L$-colorable by Thomassen~\cite{thomassen1995-34};
hence, assume that $g+s\ge 1$.
Let $G_1$ be the subgraph of $G$ induced by vertices at distance less than $200(C+5+\log(1+|S|/(C+1)))$ from $S$.
Note that if $G$ is not $L$-colorable, then it contains a subgraph $G_0$ with $S\subseteq V(G_0)$ that is $S$-critical and not $L$-colorable.
Suppose that $G_0$ has a component $Q$ that does not contain a vertex of $S$.  By Theorem~\ref{thm-distbound}, $Q$
has at most $100C$ vertices.  Since planar graphs of girth at least five are $3$-choosable, $Q$ contains a non-contractible
cycle of length at most $|V(Q)|\le 100C$, contrary to the assumptions.
Therefore, every component of $G_0$ contains a vertex of $S$, and by
Theorem~\ref{thm-distbound}, we have $G_0\subseteq G_1$.   It follows that $G$ is $L$-colorable if and only if $G_1$ is $L$-colorable.

Let $G_2$ be the graph obtained from $G_1$ by adding a new vertex adjacent to all vertices of $S$.  Note that
$G_2$ has radius at most $200(C+5+\log(1+|S|/(C+1)))=O(\log |S|)$ by Theorem~\ref{thm-distbound} and that $G_2$ can be embedded
in a surface of Euler genus at most $g+2s-2$. By Theorem~\ref{thm-eppstein}, $G_2$ (and thus also $G_1$) has tree-width at most
$K'(g+s)\log (|S|+1)$ for some absolute constant $K'$.  We apply the standard dynamic programming algorithm~\cite{listcoltw} for list-coloring
graphs from lists of bounded size, which (for lists of size at most three) has time complexity $e^{O(\tw(G_1))}|V(G_1)|)$.
Since $|V(G_1)|\le |V(G)|$, this gives the desired bound on the time complexity of the algorithm.
\end{proof}

Let us remark that in the situation of Corollary~\ref{cor-algew}, it is easy to find an $L$-coloring of $G$ when it exists in time
$O\bigl(|V(G)|^{K(g+\max(s,1))+2}\bigr)$ as follows.
We can assume that $S$ is an independent set, as edges joining vertices of the same color would prevent the existence
of an $L$-coloring, and edges joining vertices of different colors are irrelevant.
If $V(G)=S$, then the list assignment $L$ gives an $L$-coloring of $G$.
Otherwise, there exists a vertex $v\in V(G)\setminus S$ incident with a face of $F$.
Let $S'=S\cup\{v\}$, and using the algorithm of Corollary~\ref{cor-algew}, try all three possible colors for $v$ and test
whether the corresponding coloring of $S'$ extends to an $L$-coloring of $G$.  If this succeeds for at least one color $c$,
set $L(v)\colonequals\{c\}$, $S\colonequals S'$ and repeat the process until the whole graph is colored.

\section{The algorithm}\label{sec-impl}

A standard dynamic programming approach enables us to deal with short non-$F$-contractible cycles.
Let us first consider the case of a graph embedded in the sphere with precolored vertices incident with
at most two faces.

\begin{lemma}\label{lemma-algcyl}
Let $d$ be an integer such that $4\le d\le \lceil 100C_0\rceil$, where $C_0=400(10+\log 2)$.
There exists an algorithm ${\cal A}_d$ with the following specification.
The input of the algorithm is a graph $G$ of girth at least five embedded in the sphere,
a set $F$ of at most $2$ faces of $G$ such that every cycle in $G$ of length at most $d$ is $F$-contractible,
a set $S$ of vertices incident with the faces of $F$,
and a list assignment $L$ such that $|L(v)|=3$ for $v\in V(G)\setminus S$ and $|L(v)|=1$ for $v\in S$. 
The algorithm correctly decides whether $G$ is $L$-colorable.  The time complexity of the algorithm is
$O\bigl((|S|+1)^{K_0}|V(G)|\bigr)$ for some absolute constant $K_0$.
\end{lemma}
\begin{proof}
Let $K_0=2K$ for the constant $K$ from Corollary~\ref{cor-algew}.
We proceed by induction on $d$, starting from the largest value.  If $d=\lceil 100C_0\rceil$, then
we can use the algorithm of Corollary~\ref{cor-algew} as ${\cal A}_d$.
Hence, suppose that $d<\lceil 100C_0\rceil$ and
as the induction hypothesis assume that the algorithm ${\cal A}_{d+1}$ exists.  

We now describe the algorithm ${\cal A}_d$.  If $|F|\le 1$, then we use the
algorithm of Corollary~\ref{cor-algew},  since then every cycle is $F$-contractible; thus, we can assume that $F$ consists
of two faces $f_1$ and $f_2$.
If an edge joins two vertices $u,v\in S$ such that
$L(u)=L(v)$, then $G$ is not $L$-colorable.  Otherwise, $G$ is $L$-colorable if and only if $G-uv$ is $L$-colorable.
Consequently, we can also assume that $S$ forms an independent set.

Note that a cycle in $G$ is not $F$-contractible if and only if it separates $f_1$ from $f_2$.
Consider the dual $G'$ of $G$ and a maximum flow from $f_1$ to $f_2$ in $G'$ (where all edges have capacity $1$).
The size of the flow is equal to the size of a minimum cut between $f_1$ and $f_2$ in $G'$, which corresponds
to the shortest cycle separating $f_1$ from $f_2$ in $G$.  If the size of the flow is at least $d+2$ (which
can be verified in $O(|V(G)|)$ using the Ford-Fulkerson algorithm by stopping after we improve the flow $d+2$ times),
then every cycle in $G$ of length at most $d+1$ is $F$-contractible and the lemma follows by induction.
Therefore, assume that the maximum flow has size exactly $d+1$.  Then, there exists a unique $(d+1)$-cycle $Q_1$ in $G$
separating $f_1$ from $f_2$ which is nearest to $f_1$, corresponding to the cut of size $d+1$ in $G'$
bounding the set of vertices that can be reached from $f_1$ by augmenting paths.  Let $G_1$ be the graph obtained from $G$
by removing all vertices between $f_1$ and $Q_1$, including $V(f_1)$ but not $V(Q_1)$, as well as all edges of $Q_1$.
In $G_1$, we find the $(d+1)$-cycle $Q_2$ separating the face corresponding to $Q_1$ from $f_2$ that is nearest to $Q_1$,
and by removing the part between $Q_1$ and $Q_2$ as well as the edges of $Q_2$, we obtain a graph $G_2$.  We repeat this
as long as the current graph contains such a cycle.  Let $Q_1$, $Q_2$, \ldots, $Q_n$ be the sequence of $(d+1)$-cycles
obtained in this way, let $Q_0=S\cap V(f_1)$ and let $Q_{n+1}=S\cap V(f_2)$.  Note that we can inherit the flow
from the dual of $G_i$ to the dual of $G_{i+1}$ and that in order to find the cycle $Q_{i+1}$, the algorithm
visits only the vertices in $V(G_{i+1})\setminus V(G_i)$, for $1\le i\le n-1$.  Consequently we can find this sequence
of cycles in time $O(|V(G)|)$.

For $0\le i\le n$, let $H_i$ be the subgraph of $G$ between $Q_i$ and $Q_{i+1}$, not including the edges of $Q_i$ and $Q_{i+1}$.
Note that all cycles of length at most $d+1$ in $H_i$ are $\{Q_i,Q_{i+1}\}$-contractible and that for $j\in \{i,i+1\}$,
we either have $|V(Q_j)|=d+1$ or $L$ assigns a unique coloring to $Q_j$.  Furthermore, $n=O(|V(G)|)$ and only the vertices
of $Q_1\cup Q_2\cup \ldots \cup Q_n$ can belong to several of the graphs $H_i$; thus, $\sum_{i=0}^n |V(H_i)|=O(|V(G)|)$.
For $0\le i\le n$, we determine the set $\Psi_i$ of all $L$-colorings of $V(Q_i)\cup V(Q_{i+1})$ that extend
to an $L$-coloring of $H_i$, excluding those incompatible with the edges of $Q_i\cup Q_{i+1}$.
This can be done by applying the algorithm ${\cal A}_{d+1}$ to $H_i$ repeatedly, precoloring the vertices of $Q_i\cup Q_{i+1}$
in all possible ways (the number of $L$-colorings of $Q_i\cup Q_{i+1}$ is bounded by the constant $3^{2d+2}$).  The total time spend by this is $O\bigl((|S|+1)^{K_0}|V(G)|\bigr)$.

Finally, for $1\le i\le n+1$, we determine the set $\Phi_i$ of all $L$-colorings of $V(Q_0)\cup V(Q_i)$ that extend to an $L$-coloring
of the subgraph of $G$ between $Q_0$ and $Q_i$ (inclusive)---we already know the set $\Phi_1=\Psi_0$, and $\Phi_{i+1}$ can be obtained
by combining $\Phi_i$ with $\Psi_i$ in constant time.  The time to determine $\Phi_{n+1}$ is thus $O(|V(G)|)$.
The graph $G$ is $L$-colorable if and only if $\Phi_{n+1}$ is nonempty.
\end{proof}

Let us note that if $d=4$, then the assumption on $F$-contractibility of all cycles of length at most $d$
is void, since the graph has girth at least five; hence, we can drop the assumption entirely.

A similar argument is used to give the algorithm in the general case.
Given a graph $G$ embedded in a surface $\Sigma$ with boundary, a non-contractible cycle $C$ in $G$
is \emph{almost contractible} if there exists a component $b$ of the boundary of $\Sigma$
such that $C$ is contractible in the surface obtained from $\Sigma$ by capping $b$ with a disk;
in this case, we say that $C$ \emph{surrounds} $b$.

\begin{theorem}\label{thm-main}
Let $g,s\ge 0$ be fixed integers.
There exists an algorithm with the following specification.
The input of the algorithm is a graph $G$ of girth at least five embedded in a surface $\Sigma$ of Euler genus at most $g$ and without boundary,
a set $F$ of at most $s$ faces of $G$,
a set $S$ of vertices incident with the faces of $F$,
and a list assignment $L$ such that $|L(v)|=3$ for $v\in V(G)\setminus S$ and $|L(v)|=1$ for $v\in S$.
The algorithm correctly decides whether $G$ is $L$-colorable.  The time complexity of the algorithm is
$O\bigl((|S|+1)^{K(g+s)}|V(G)|\bigr)$ for some absolute constant $K$.
\end{theorem}
\begin{proof}
We proceed by induction, assuming that the theorem holds for all surfaces of Euler genus less than $g$,
and for all graphs embedded in a surface of Euler genus $g$ with precolored vertices incident with fewer than $s$ faces.
If $g=0$ and $s\le 2$, then we can use the algorithm ${\cal A}_4$ of Lemma~\ref{lemma-algcyl}.
Therefore, assume that $g>0$ or $s>2$.
Let $C=400(2g+s-1)(10+\log (2g+s))$.
As usual, we can assume that $S$ forms an independent set and that $G$ is connected and its embedding is $2$-cell.

Let us drill a hole in each face of $F$, obtaining a surface $\Sigma'$.  We use the algorithm (inspired by the result of Cabello and Mohar~\cite{cabello-sepcurv})
described in Dvořák et al.~\cite{coltrfree} before the proof of Theorem 8.3 
to test whether the embedding of $G$ in $\Sigma'$ contains a non-contractible cycle of length less than $100C$ that is not almost contractible,
in time $O(|V(G)|)$.  Suppose that $Q$ is such a cycle.  For each $L$-coloring $\psi$ of $Q$, let $L^{\psi}$ be the list assignment such that
$L^{\psi}(v)=\{\psi(v)\}$ for $v\in V(Q)$ and $L^{\psi}(v)=L(v)$ for $v\in V(G)\setminus V(Q)$.  There are only constantly many choices for $\psi$,
and $G$ is $L$-colorable if and only if it is $L^{\psi}$-colorable for some $L$-coloring $\psi$ of $Q$.
Note that $G$ is $L^{\psi}$-colorable if and only if $G_Q$ is $L^{\psi}_Q$-colorable, and each component of the
surface $\widehat{\Sigma_Q}$ gives rise to a simpler instance of the problem (either the component has smaller Euler genus than $\Sigma$,
or it is homeomorphic to $\Sigma$, but the precolored vertices in $G_Q$ are incident with fewer than $|F|$ faces in this component).
The claim of the theorem follows by induction applied to each of the components (and all the possible colorings $\psi$).

Therefore, we can assume that all non-contractible cycles of length at most $100C$ in the embedding of $G$ on $\Sigma'$ are almost
contractible.  Since $g>0$ or $s>2$, the component of the boundary of $\Sigma'$ surrounded by such a cycle is unique.
The closer inspection of the algorithm of Dvořák et al.~\cite{coltrfree} shows that it actually implements a data structure representing
the graph $G$ that can be initialized in time $O(|V(G)|)$ and supports the following operations in a constant time:
\begin{itemize}
\item Remove an edge or an isolated vertex.
\item Decide whether a vertex belongs to a non-contractible cycle of length at most $100C$, and if that is the case,
return such a cycle; furthermore, if the cycle is almost contractible, it decides which component of the boundary
it surrounds.
\end{itemize}
Using this data structure, we process the vertices of $G$ one by one, testing whether they belong to an almost contractible cycle $Q$
of length at most $100C$.   If that is the case, we remove all vertices and edges of $G$ between $Q$ and the surrounded component $b$ of the boundary
(including the edges, but not the vertices, of $Q$).  The part of graph to be removed can be found by a depth-first search from $b$ in time
proportional to the size of the removed part.  We then continue the process with the remaining vertices of $G$.  We end up with
a subgraph $G'$ of $G$.  For each component $b$ of the boundary of $\Sigma'$, let $S_b$ denote the subset of $S$ incident with the face of $G$
corresponding to $b$.  Let $G_b$ denote the removed part of $G$ between $b$ and $G'$; if nonempty,
$V(G_b)\cap V(G')$ is a vertex set of a cycle $Q_b$ in $G_b$ of length at most $100C$ that surrounds $b$ in the embedding of $G$ in $\Sigma'$.
For each $b$, we determine the set of all $L$-colorings of $S_b\cup V(Q_b)$ that extend to an $L$-coloring of $G_b$ using
the algorithm ${\cal A}_4$ of Lemma~\ref{lemma-algcyl}.
Furthermore, note that the embedding of $G'$ in $\Sigma'$ contains no non-contractible cycle of length at most $100C$, thus
we can determine the set of all $L$-colorings of its intersection with $S$ and with the cycles $Q_b$ using the algorithm of Corollary~\ref{cor-algew}.
Combining these sets (whose size is bounded by a constant depending on $g$ and $s$), we can decide whether $G$ is $L$-colorable.
Observe that this algorithm has the required time complexity.
\end{proof}

Also, a straightforward cutting argument enables us
to deal with precolored vertices being contained in connected subgraphs instead of incident with a common face.

\begin{corollary}\label{cor-main}
Let $g,s\ge 0$ be fixed integers.
There exists an algorithm with the following specification.
The input of the algorithm is a graph $G$ of girth at least five embedded in a surface $\Sigma$ of Euler genus at most $g$ and without boundary,
a subgraph $Q$ of $G$ with at most $s$ components,
and a list assignment $L$ such that $|L(v)|=3$ for $v\in V(G)\setminus V(Q)$ and $|L(v)|=1$ for $v\in V(Q)$. 
The algorithm correctly decides whether $G$ is $L$-colorable.  The time complexity of the algorithm is
$O\bigl((|V(Q)|+1)^{K(g+s)}|V(G)|\bigr)$ for some absolute constant $K$.
\end{corollary}
\begin{proof}
We can assume that $Q$ is a forest and each component of $Q$ has at least three edges (otherwise, remove
or add edges to $Q$ to make these conditions hold).  Consequently, $\widehat{\Sigma_Q}$ is homeomorphic to $\Sigma$,
$G_Q$ has girth at least five and each component
of $Q$ corresponds to a cycle in $G_Q$ bounding a face.
Let $S=V(Q_Q)$ and note that $|S|=2|E(Q)|< 2|V(Q)|$.
Note also that $G$ is $L$-colorable if and only if $G_Q$ is $L_Q$-colorable,
since all vertices of $Q$ have lists of size one.  The claim then follows by Theorem~\ref{thm-main}.
\end{proof}

Again, by coloring the vertices of $G$ one by one, we can actually find an $L$-coloring
of $G$ when it exists, in time $O\bigl(|V(G)|^{K(g+\max(s,1))+2}\bigr)$.

\section{Choosability}\label{sec-choos}

Here, we extend the algorithm of Theorem~\ref{thm-main} to the case of $3$-choosability.
The basic ingredients are the following generalizations of Corollary~\ref{cor-algew} and Lemma~\ref{lemma-algcyl}.
Let $S$ be a set of vertices.  We say that two list assignments for $S$ are are \emph{equivalent} when
they differ only by renaming the colors.  Let $\CC(S)$ denote a maximal set of pairwise non-eqivalent
assignments of lists of size three to $S$, and note that $\CC(S)$ is finite.  If $G$ is a graph with $S\subseteq V(G)$,
then let $\CC(S,G)$ denote the set consisting of all pairs $(L_0,\Psi)$ such that $L_0\in\CC(S)$,
$\Psi$ is a set of $L_0$-colorings of $S$ and there exists a list assignment $L$ for $G$ such that $|L(v)|=3$ for $v\in V(G)\setminus S$,
$L(v)=L_0(v)$ for $v\in S$ and an $L_0$-coloring $\psi$ of $S$ extends to an $L$-coloring of $G$ if and only if $\psi\in \Psi$.

\begin{corollary}\label{cor-algewchoos}
Let $g,s\ge 0$ be fixed integers.  Let $C=0$ if $g=s=0$ and $C=400(2g+s-1)(10+\log (2g+s))$ otherwise.
There exists a function $f_{g,s}$ and an algorithm with the following specification.
The input of the algorithm is a graph $G$ of girth at least five embedded in a surface $\Sigma$ of Euler genus at most $g$ and without boundary,
a set $F$ of at most $s$ faces of $G$ such that every cycle in $G$ of length at most $100C$ is $F$-contractible and
a set $S$ of vertices incident with the faces of $F$.  The algorithm outputs the set $\CC(S,G)$ in time $O(f_{g,s}(|S|)|V(G)|)$.
\end{corollary}
\begin{proof}
If $g=s=0$, then $G$ is $3$-choosable by Thomassen~\cite{thomassen1995-34}, and $\CC(S,G)$ consists of the pair $(L_0,\{\psi_0\})$, where $L_0$
is the null list assignment and $\psi_0$ is the null coloring.  Hence, assume that $g+s\ge 1$.  Let $G_1$ be the subgraph of $G$ induced by vertices at distance less than $200(C+5+\log(1+|S|/(C+1)))$ from $S$.
Consider any list assignment $L$ such that $|L(v)|=3$ for $v\in V(G)\setminus S$ and $|L(v)|=1$ for $v\in S$.
As in Corollary~\ref{cor-algew}, we have that $G$ is $L$-colorable if and only if $G_1$ is $L$-colorable.
Consequently, $\CC(S,G)=\CC(S,G_1)$, and since $G_1$ has bounded treewidth, we can determine $\CC(S,G_1)$ in linear time
by a standard dynamic programming algorithm.
\end{proof}

\begin{lemma}\label{lemma-algcylchoos}
Let $d$ be an integer such that $4\le d\le \lceil 100C_0\rceil$, where $C_0=400(10+\log 2)$.
There exists a function $f_d$ and an algorithm ${\cal A}'_d$ with the following specification.
The input of the algorithm is a graph $G$ of girth at least five embedded in the sphere,
a set $F$ of at most $2$ faces of $G$ such that every cycle in $G$ of length at most $d$ is $F$-contractible and
a set $S$ of vertices incident with the faces of $F$.
The algorithm outputs a set $\CC(S,G)$ in time $O(f_d(|S|)|V(G)|)$.
\end{lemma}
\begin{proof}
As in Lemma~\ref{lemma-algcyl}, we proceed by induction on $d$, starting from the largest value (we use
Corollary~\ref{cor-algewchoos} as the basic case).  Hence, suppose that $d<\lceil 100C_0\rceil$ and
as the induction hypothesis assume that the algorithm ${\cal A}'_{d+1}$ exists.
We find a maximal sequence $Q_1$, $Q_2$, \ldots, $Q_n$ of non-crossing $(d+1)$-cycles separating the faces $f_1$ and $f_2$ of $F$
in the same way as in Lemma~\ref{lemma-algcyl} and let $Q_0=S\cap V(f_1)$ and $Q_{n+1}=S\cap V(f_2)$.
For $0\le i\le n$, we then apply ${\cal A}'_{d+1}$ to the subgraph $H_i$ of $G$ between $Q_i$ and $Q_{i+1}$, obtaining
$\CC(V(Q_i)\cup V(Q_{i+1}), H_i)$.  Let $G_k=\bigcup_{0\le i\le k} H_i$.  We now determine $\CC(V(Q_0)\cup V(Q_{k+1}), G_k)$
for $0\le k\le n$ by induction on $k$.  Since $G_0=H_0$, we can assume that $k\ge 1$ and that $\CC_1=\CC(V(Q_0)\cup V(Q_k), G_{k-1})$ is already known.
Let $\CC_2=\CC(V(Q_k)\cup V(Q_{k+1}), H_k)$.  By combining $\CC_1$ with $\CC_2$, we determine $\CC_3=\CC(V(Q_0)\cup V(Q_k)\cup V(Q_{k+1}), G_k)$---a pair $(L_0,\Psi)$
belongs to $\CC_3$ if and only if there exist $(L'_0,\Psi')\in \CC_1$ and $(L''_0,\Psi'')\in \CC_2$ such that $L'_0$ and $L''_0$ are restrictions of $L_0$
to the respective sets and $\Psi$ consists of colorings $\psi$ such that $\Psi'$ and $\Psi''$ contain the restriction of $\psi$ to the respective sets.
The set $\CC(V(Q_0)\cup V(Q_{k+1}), G_k)$ is obtained by including all pairs $(L_0,\Psi)$ such that there exists $(L'_0,\Psi')\in\CC_3$
such that $L'_0$ extends $L_0$ and $\Psi$ consists of $L_0$-colorings $\psi$ such that an extension of $\psi$ is contained in $\Psi'$.
In the end, we have $\CC(S,G)=\CC(V(Q_0)\cup V(Q_{n+1}), G_n)$.
\end{proof}

Now, to decide whether a graph $G$ of girth at least five embedded in a fixed surface of genus $g$ is $3$-choosable, we
cut $G$ along short cycles as in the proof of Theorem~\ref{thm-main}.  This way, we express
$G$ as an edge-disjoint union of graphs $G_1$, \ldots, $G_k$, where $k$ is bounded by a function of $g$ and
the set $S$ of vertices contained in at least two of these graphs has size bounded by a function of $g$.
Furthermore, for $1\le i\le k$, the graph $G_i$ with the set $S\cap V(G_i)$ satisfies either assumptions of Corollary~\ref{cor-algewchoos}
or of Lemma~\ref{lemma-algcylchoos} with $d=4$.  Consequently, we can determine $\CC(S,G)$ in linear time.  
The graph $G$ is $3$-choosable if and only if $\Psi\neq\emptyset$ for all $(L_0,\Psi)\in\CC(S,G)$.

\section{Concluding remarks}

The degree of the polynomial bounding the time complexity of our algorithm depends
mostly on the bound on the distance given by Theorem~\ref{thm-distbound}.
The algorithm is easy to implement, especially in the planar case.
While the bound in Theorem~\ref{thm-distbound} is rather high, it seems likely that it
is possible to reduce it significantly.  This could make the algorithm of Corollary~\ref{cor-main} practical, at least
for the special case of a planar graph with a connected precolored subgraph.

Nevertheless, an interesting open question is whether we can eliminate the dependence of the exponent
on the genus entirely, and thus obtain an FPT algorithm for the problem of extension of a precoloring
of a connected subgraph of a graph embedded in a fixed surface.  The same remark holds for the
case of finding a coloring of a graph without precolored subgraph, where we currently have
to introduce a precolored subgraph whose size may be up to $\Omega(\sqrt{|V(G)|)}$.

\bibliographystyle{acm}
\bibliography{3choos}
\end{document}